\documentclass[12pt]{article}

\usepackage{amsmath}
\usepackage{amsfonts}
\usepackage{amsthm}

\newtheorem{thm}{Theorem}
\newtheorem{lem}[thm]{Lemma}
\newtheorem{defi}[thm]{Definition}
\newtheorem{prop}[thm]{Proposition}
\newtheorem{cor}[thm]{Corollary}

\newtheorem{exam}[thm]{Example}
\newtheorem{rem}[thm]{Remark}

\newcommand{\N}{\mathbb{N}}

\newcommand{\R}{\mathbb{R}}

\newcommand{\C}{\mathbb{C}}

\newcommand{\cL}{\mathcal{L}}
\newcommand{\cG}{\mathcal{G}}

\newcommand{\cP}{\mathcal{P}}

\newcommand{\E}{\mathbb{E}}

\newcommand{\sw}{\mathcal{S}(\mathbb{R})}
\newcommand{\td}{\mathcal{S}'(\mathbb{R})}
\newcommand{\leb}{L^2(\mathbb{R})}

\newcommand{\si}{\sigma}

\newcommand{\gam}{\gamma}
\newcommand{\al}{\alpha}

\newcommand{\om}{\omega}
\newcommand{\halb}{\frac{1}{2}}
\newcommand{\de}{\delta}
\newcommand{\var}{\mathrm{var}}
\newcommand{\vi}{\varphi}
\newcommand{\ep}{\varepsilon}
\newcommand{\wick}{\diamond}

\newcommand{\lb}{\left(}

\newcommand{\lcb}{\left\{}
\newcommand{\lvb}{\left|}
\newcommand{\lVb}{\left\|}

\newcommand{\lab}{\left\langle}
\newcommand{\lAb}{\left\langle\langle}
\newcommand{\rb}{\right)}

\newcommand{\rcb}{\right\}}
\newcommand{\rvb}{\right|}
\newcommand{\rVb}{\right\|}

\newcommand{\rab}{\right\rangle}
\newcommand{\rAb}{\right\rangle\rangle}
\newcommand{\1}{\mathbf{1}}

\title{A White Noise Approach to the Feynman Integrand for Electrons in Random Media}

\author{M. Grothaus, F. Riemann and H. P. Suryawan\\University of Kaiserslautern, Germany}

\date{}

\begin{document}

\maketitle

\begin{abstract}
\noindent Using the Feynman path integral representation of quantum mechanics it is possible to derive a model of an electron in a random system containing dense and weakly-coupled scatterers, see \cite{EG64}. The main goal of this paper is to give a mathematically rigorous realization of the corresponding Feynman integrand in dimension one based on the theory of white noise analysis. We refine and apply a Wick formula for the product of a square-integrable function with Donsker's delta functions and use a method of complex scaling. As an essential part of the proof we also establish the existence of the exponential of the self-intersection local times of a one-dimensional Brownian bridge. As result we obtain a neat formula for the propagator with identical start and end point. Thus, we obtain a well-defined mathematical object which is used to calculate the density of states, see e.g. \cite{EG64}.
\end{abstract}

\maketitle

\begin{center}
{\small Copyright 2014 American Institute of Physics. This article may be downloaded for personal use only. Any other use requires prior permission of the author and the American Institute of Physics.

The following article appeared in Journal of Mathematical Physics 55, Issue 1 and may be found at:}

{\footnotesize\tt http://scitation.aip.org/content/aip/journal/jmp/55/1/10.1063/1.4862744}
{\footnotesize\tt http://dx.doi.org/10.1063/1.4862744}
\end{center}

\section{Introduction}
We start with a motivation from Physics. Using Feynman's path integral approach to quantum mechanics, see for example \cite{FeHi65}, Edwards and Gulyaev in 1964 first introduced a model of an electron moving in a random medium containing dense and weakly-coupled scatterers (e.g.~impurities) for the purpose of investigating the nature of electronic states in a disordered system. Below we briefly sketch the heuristic model proposed in \cite{EG64} and developed in \cite{Sam74}. We consider an electron moving in a set of $N$ rigid scatterers, confined within a volume $V\subset \R^d$, $d=1,2,3$, with positive Lebesgue measure $0<dx(V)<\infty$ and having a density $\rho=\frac{N}{dx(V)}$. Such a system is described by the Hamiltonian
\[ H=-\frac{\hbar^2}{2m}\Delta +\sum_{j=1}^N\eta v(q-r_j), \]
where $m$ is the mass of the electron, $\hbar$ is the reduced Planck constant, $\Delta$ is the Laplacian, $q$ is the multiplication operator, $v(\cdot-r_j)$ represents the potential of a single scatterer at position $r_j\in \R^d$, $1\le j\le N$, and the non-negative parameter $\eta$ measures the strength of the interaction of one scatterer. The propagator $G_{N,V}^v$ of such a system can be expressed in the path integral formalism as
\[
G_{N,V}^v=\int \exp\lb \frac{i}{\hbar}\int_0^T \lb \frac{m}{2}\dot{x}(t)^2-\sum_{j=1}^N\eta v(x(t)-r_j)\rb \, dt \rb \, \, \mathcal{D}(x) ,
\]
where $\mathcal{D}(x)$ denotes Feynman's path measure, i.e., integration over all paths $x:[0,T]\to \R^d$, $0<T<\infty$, with start point $x(0)=x_0\in \R^d$ and end point $x(T)=x_T\in \R^d$. After averaging over all possibilities of scatterers' configurations and taking the thermodynamic limit $\lim_{N,dx(V)\to \infty}\frac{N}{dx(V)}=\rho<\infty$ the propagator for an electron in random media containing dense and weakly-coupled scatterers such that $\lim_{\rho \to \infty,\eta \to 0}\rho \eta^2=k<\infty$ can be written informally as
\[
G_W=\int \exp\lb \frac{i}{\hbar}S_W(x) \rb \, \mathcal{D}(x),
\]
where the classical action as a function of the path $x$ is given by
\[
S_W(x)=\int_0^T\frac{m}{2}\dot{x}(t)^2 \, dt +\frac{ik}{2\hbar}\int_0^T\int_0^T W(x(t)-x(s))\, ds \, dt. 
\]
In many situations of physical study the electron-scatterers potential is taken to be Gaussian function
\[
v=(\pi l^2)^{-d/2}\exp\lb -\frac{|\cdot|^2}{l^2}\rb, \quad 0<l<\infty, 
\]
and yields the correlation function
\begin{equation}\label{potent}
W(x(t)-x(s))=(\pi L^2)^{-d/2}\exp\lb -\frac{|x(t)-x(s)|^2}{L^2}\rb, \quad t,s\in [0,T] ,
\end{equation}
where the correlation length $L$ satisfies $L^2=2l^2$ and $\lvb \cdot \rvb$ denotes the Euclidean norm on $\R^d$. The exact propagator for this case has been obtained explicitly by using a finite-dimensional approximation, see e.g. \cite{KL86}.

Motivated by the Edwards' model discussed above we are interested in the investigation of the Feynman path integral with classical action containing the correlation function  as in (\ref{potent}) for the limiting case $L\to 0$ of the correlation length of the electron-scatterers interaction system. More precisely, for a Gaussian scattering potential and by letting $L\to 0$ in (\ref{potent}) we obtain the propagator for electron-scatterers interaction with non-local Dirac delta action
\[ G_{\de}=\int \exp\lb \frac{i}{\hbar}S_{\de}(x)\rb \, \mathcal{D}(x) , \]
where
\[ S_{\de}(x)= \int_0^T \lb \frac{m}{2}\dot{x}(t)^2 + \frac{ik}{2\hbar} \int_0^T\de(x(t)-x(s))\, ds\rb \, dt. \]
It is clear that these expressions are only informal. The main goal in this paper is to give a mathematically sound realization for them.
There have been many approaches for giving a mathematically rigorous meaning to the Feynman path integral e.g. by using operator semigroup theory, analytic continuation or infinite dimensional oscillatory integrals, see \cite{AHM08} and references therein for a comprehensive discussion. In this paper we choose a white noise approach. White noise analysis is a mathematical framework which offers generalizations of concepts from finite-dimensional analysis, like differential operators, Fourier transform and distribution theory to an infinite-dimensional setting. For a complete account on this theory including its huge range of applications we refer to \cite{HKPS93,Kuo96,Oba94}. The idea of realizing Feynman integrals within the white noise framework was first mentioned in the work of Hida and Streit \cite{HS83}. We should emphasize that the white noise approach to the Feynman path integral has some interesting features, for example the admissible potentials may be very singular. In addition, instead of giving meaning directly to the Feynman integral we define the Feynman integrand as a white noise distribution. By taking the generalized expectation with respect to the white noise measure we obtain the propagator. For the development and results of the Feynman path integral within white noise analysis framework see for example \cite{SS04,Vog10,Wes95} and references therein. We summarize our strategies and results as follows. As a starting point we informally express the Feynman integrand for electrons in random media with Dirac delta correlation function without kinetic energy part as
\begin{equation}\label{motiv}
\exp\lb \int_0^T \lb -\frac{ik}{2\hbar} \int_0^T \de\lb x(t)-x(s)\rb ds\rb dt \rb \cdot \de_{x_T}\lb x(T)\rb .
\end{equation}
For taking into account also the kinetic energy part we scale by $\sqrt{i}$ and obtain the Feynman-Kac-Cameron-Doss integrand
\begin{equation}\label{Doss}
\exp\lb \frac{1}{i} \int_0^T \lb -\frac{ik}{2\hbar} \int_0^T \de\lb \sqrt{i}(B_t-B_s)\rb ds\rb dt \rb \cdot \de_{x_T}\lb x_0+\sqrt{i}B_T\rb ,
\end{equation}
where $(B_t)_{t\in [0,T]}$ is a one-dimensional standard Brownian motion (starting in 0 at 0). This ansatz is motivated by the complex scaling method in the sense of Cameron-Doss on the stochastic representation of a solution of heat equations (Feynman-Kac formula). Doss proved that for a class of potentials satisfying some analyticity and integrability conditions, the complex scaling approach as in (\ref{Doss}) is equivalent to the classical Feynman path integral formulation. This means scaling of Brownian motion by $\sqrt{i}$ gives the kinetic energy term in the context of the white noise or Wiener measure, respectively. For details and proofs see \cite{Doss82}. Since we are not dealing with a potential from the Doss class, we take (\ref{Doss}) as our starting point and give meaning to the product in (\ref{Doss}). More precisely, we prove that (\ref{Doss}) is a well-defined object as a white noise distribution. It is also important to note that the object in the exponential term leads to the so-called self-intersection local time of a Brownian motion.

The present paper is organized as follows: In Section 2 we further develop some tools from the theory of white noise analysis. In particular, we refine a Wick formula which enables us to multiply a class of square-integrable functions with Donsker's delta functions. For this purpose we do careful analysis on projection operators acting on white noise functionals. Section 3 is devoted to the study of self-intersection local times of a Brownian bridge. This is needed for applying the Wick formula to the product in (\ref{Doss}). Using an approximation procedure we show that
\[ \exp\lb z\int_0^T\int_0^T\de\lb X_t-X_s\rb \, ds \, dt\rb,  \quad z\in \C, \operatorname{Re}z\le 0, \]
where $(X_t)_{t\in [0,T]}$ is a Brownian bridge, is a square-integrable function. Here we have to restrict ourselves to the case $d=1$. Several ideas for proving this result we got from \cite{HN05}. In Section 4 we apply the results that were established in the previous sections to our main problem described above. Indeed, we are able to give a mathematically rigorous meaning to (\ref{Doss}) as a regular generalized function from $\cG'$, see e.g. \cite{PT95}. We emphasize that the Wick formula improved in Section 2 enables us to represent the pointwise product (\ref{Doss}) in terms of the Wick product which is generally well-defined for elements from $\cG'$. We also obtain a neat formula for the propagator for the electrons in random media with Dirac delta correlation function with identical start and end point. Thus, we obtain a well-defined mathematical object which is used to calculate the density of states, see e.g. \cite{EG64, Sam74, KL86}.

\section{White Noise Analysis}

In this section we briefly recall the concepts and results of white noise analysis used throughout this work, for a detailed explanation see e.g. \cite{HKPS93,Kuo96,Oba94}.

\subsection{White Noise Measure}

Let $L^2(\R)$ denote the space of real-valued square-integrable functions with respect to the Lebesgue measure on $\R$ equipped with its usual inner product $(\cdot,\cdot)$ and corresponding norm $|\cdot|$. The Schwartz space of rapidly decreasing functions on $\R$  is denoted by $\sw$ and equipped with its usual nuclear topology. Its topological dual space is the space of tempered distributions and denoted by $\td$. By identifying $\leb$ with its dual space, the dual pairing $\lab \cdot ,\cdot\rab$ between $\td$ and $\sw$ is realized as an extension of the inner product in $\leb$, i.e., $\lab \xi ,\zeta\rab =(\xi,\zeta)$ for $\xi\in \leb$ and $\zeta \in \sw$. Hence we obtain the Gel'fand triple $\sw \subset \leb \subset \td$. Equipped with its cylindrical $\si$-algebra $\mathcal{C}$, the standard Gaussian measure (or the white noise measure) $\mu$ on $\td$ arises from its characteristic function via the Bochner-Minlos theorem by
\[ \int_{\td}\exp\lb i\lab \om ,\xi\rab \rb \, d\mu(\om)=\exp\lb -\halb \lab \xi,\xi\rab \rb ,\quad \xi \in \sw. \]
The probability space $(\td, \mathcal{C}, \mu)$ is called the white noise space. For $1\le p\le \infty$ we abbreviate $L^p(\mu):=L^p(\td,\mu;\C)$, the space of complex-valued $p$-integrable functions with respect to $\mu$, together with its usual norm $\lVb \cdot\rVb_{L^p(\mu)}$. A fundamental property of the measure $\mu$ is that for fixed $\xi_1,\ldots,\xi_d\in \sw$, $d\in \N$, the random vector $\lb \lab \cdot,\xi_1 \rab,\ldots, \lab \cdot,\xi_d \rab\rb$ is centered Gaussian with covariance structure $\lb \lab \xi_k,\xi_l \rab\rb_{k,l=1,\ldots,d}$. Thus, if we extend $\lab \cdot,\cdot \rab$ in a bilinear way to elements from the complexified spaces, we obtain that the space of smooth polynomials
\[ \cP:=\mathrm{span}\lcb \lab \cdot,\xi \rab^n : \xi \in \mathcal{S}(\R;\C) , n\in \N_0\rcb \]
is a subspace of $L^2(\mu)$. We use the notation $L^2(\R^n;\C)_{sym}$ for the symmetric Hilbert space of complex-valued square-integrable functions with respect to the Lebesgue measure and keep the symbol $\lab\cdot,\cdot\rab$ for its bilinear dual pairing and $|\cdot|$ for its norm. Similar as before, the dual pairing between $\mathcal{S}'(\R^n)_{sym}$ and $\mathcal{S}(\R^n;\C)_{sym}$ is realized as a bilinear extension of $\lab \cdot,\cdot \rab$ and is denoted by the same symbol. With this notation, each $\vi\in \cP$ of degree $N\in \N_0$ can uniquely be represented as a Wick polynomial
\begin{equation}\label{Wickpol}
\vi(\om) =\sum_{n=0}^N \lab :\om^{\otimes n}:, \vi^{(n)}\rab, \quad \vi^{(n)} \in \mathrm{span}\lcb \xi^{\otimes n}:\xi \in \mathcal{S}(\R;\C)\rcb, \, n\in \N_0 , 
\end{equation}
where $:\om^{\otimes n}:\in \mathcal{S}'(\R^n)_{sym}$ denotes the $n$-th Wick power of $\om \in \td$ and $\vi^{(n)}$ is called the $n$-th kernel of $\vi$. They have the advantage to fulfill the orthogonality relation
\[ \int_{\td}\lab :\om^{\otimes n}:, \vi^{(n)}\rab \lab :\om^{\otimes m}:, \psi^{(m)}\rab \, d\mu(\om)=\de_{nm}n! \lab \vi^{(n)},\psi^{(m)} \rab , \,\,\, n,m\in \N_0, \]
where $\de_{nm}$ denotes the Kronecker delta. This implies that for general $f^{(n)}\in L^2(\R^n;\C)_{sym}$ we can define $\lab :\cdot^{\otimes n}:, f^{(n)}\rab$ as an $L^2(\mu)$-limit. As an example a Brownian motion $(B_t)_{t\in [0,T]}$ starting in $0$ at time $0$ can be realized within this framework by
\[ B_t:=\lab \cdot, \1_{[0,t)}\rab , \]
where $\1_A$ denotes the indicator function of the set $A\subset \R$. The Kolmogorov-Chentsov theorem ensures that there exists a modification of $(B_t)_{t\in [0,T]}$ which has continuous paths almost surely. From now on we always work with a standard Brownian motion, i.e. its continuous modification starting in $0$. 

By density of the space of polynomials, for every $F\in L^2(\mu)$ there exists a unique sequence $(f^{(n)})_{n\in \N_0}$ where $f^{(n)}\in L^2(\R^n;\C)_{sym}$ such that
\[ F=\sum_{n=0}^{\infty}\lab :\cdot^{\otimes n}:, f^{(n)}\rab , \]
where the convergence holds in $L^2(\mu)$. This expansion is called (Wiener-It\^o) chaos decomposition, $f^{(n)}$ is called the $n$-th kernel of $F$, and $\| F\|_{L^2(\mu)}^2=\sum_{n=0}^{\infty}n! \lvb f^{(n)} \rvb^2$.

\subsection{Regular Test Functions and Distributions}

In this paper the space $\cG$ of regular test functions and its dual space $\cG'$ of regular distributions are of interest. They were first introduced and analyzed in \cite{PT95} and later characterized via the Bargmann-Segal transform in \cite{GKS97}. An important example of a regular distribution is the Donsker's delta function. It is also important to note that pointwise multiplication is a continuous operation from $\cG \times \cG$ to $\cG$. The space $\cG$ is the subspace of $L^2(\mu)$ consisting of all
\begin{equation}\label{regtes}
\vi =\sum_{n=0}^{\infty}\lab :\cdot^{\otimes n}:, \vi^{(n)}\rab, \quad \vi^{(n)}\in L^2(\R^n;\C)_{sym}, \,  n\in \N_0,
\end{equation}
such that
\[ \lVb \vi \rVb_q^2:=\sum_{n=0}^{\infty} n!2^{q n}\lvb \vi^{(n)}\rvb^2 <\infty \]
for every $q\in \N_0$, and the family of norms $\|\cdot\|_q$ is taken to topologize $\cG$, i.e. a sequence $(\vi_k)_{k\in \N}$ converges in $\cG$ if and only if it converges with respect to each of the norms. Obviously $\cP\subset \cG$, and since $\|\cdot\|_0=\|\cdot\|_{L^2(\mu)}$ we have that $\cG \subset L^2(\mu)$ continuously. Similar as before, the dual pairing $\lAb \cdot,\cdot\rAb$ between $\cG'$ and $\cG$ is realized as the bilinear extension of the inner product on the real part of $L^2(\mu)$ and we obtain the triple $\cG \subset L^2(\mu) \subset \cG'$. More generally, it has been shown in \cite{PT95} that
\[ \cG \subset \bigcap_{1\le p<\infty}L^p(\mu) \quad \mbox{ and } \quad \bigcup_{1<p\le \infty}L^p(\mu) \subset \cG' \]
continuously with respect to the projective and inductive limit topology, respectively. 

Important examples of elements in $\cG$ are the Wick exponentials
\[ :\exp\lb \lab \cdot, \xi\rab\rb: \, :=\exp\lb \lab \cdot, \xi\rab-\halb \lab \xi, \xi\rab \rb =\sum_{n=0}^{\infty} \frac{1}{n!}\lab :\cdot^{\otimes n}:, \xi^{\otimes n}\rab, \quad \xi \in L^2(\R;\C). \]
The $S$-transform of $\Phi \in \cG'$ is defined to be a mapping $S\Phi$ given by
\[ L^2(\R;\C) \ni \xi \mapsto S\Phi(\xi):=\lAb \Phi, :\exp\lb \lab \cdot, \xi\rab\rb:\rAb \in \C .  \]
The Wick exponentials form a total set in $\cG$, so each $\Phi \in \cG'$ is uniquely characterized by its $S$-transform. Since $\1 \in \cG$, the generalized expectation of $\Phi \in \cG'$ can be defined to be $\E_{\mu}(\Phi) :=\lAb \Phi ,\1 \rAb =S\Phi (0)$.
For $\Phi,\Psi\in \cG'$ their Wick product $\Phi \wick \Psi$ is defined to be the unique element in $\cG'$ such that $S(\Phi \wick \Psi)(\xi)=S\Phi(\xi) \cdot S\Psi(\xi)$ holds for all $\xi \in L^2(\R;\C)$. It is important to note that $\wick$ is continuous from $\cG' \times \cG'$ to $\cG'$. As mentioned before, $\cG$ is closed under pointwise multiplication which is a continuous operation from $\cG \times \cG$ to $\cG$. Hence, one can extend this multiplication allowing one factor to be in $\cG'$ by defining
\[ \lAb \Phi \cdot \vi ,\psi\rAb :=\lAb \Phi, \vi \cdot \psi \rAb, \quad \Phi \in \cG', \vi,\psi \in \cG, \]
and this multiplication is a continuous operation from $\cG'\times \cG$ to $\cG'$, see \cite{PT95}.

A well-established regular distribution is Donsker's delta $\de\lb \lab \cdot,\eta\rab -a\rb$ which is defined for $a\in \C$ and $\eta \in L^2(\R;\C)$ with $\lab \eta ,\eta\rab\notin (-\infty,0]$ and characterized via its $S$-transform
\[ S\lb \de\lb \lab \cdot,\eta\rab -a\rb\rb (\xi)=\frac{1}{\sqrt{2\pi \lab \eta,\eta\rab}} \exp\lb -\frac{1}{2\lab \eta,\eta\rab}\lb a-\lab \xi,\eta\rab\rb^2 \rb, \, \xi \in L^2(\R;\C) .\]
In applications, for example in the context of Feynman integrals, a common choice is $\eta :=\1_{[0,t)}$, $t>0$, and $a\in \R$. Hence Donker's delta function can be considered as the informal composition of the Dirac delta distribution $\de_a\in \td$ with Brownian motion. In that case Donsker's delta serves to pin a Brownian motion path at time $t$ in the point $a$. It can also be proved, using uniqueness of the $S$-transform, that Donsker's delta function is homogeneous of degree $-1$ for $z\in \C$ with $\operatorname{arg}z\in \lb -\frac{\pi}{2},\frac{\pi}{2} \rb$, i.e.,
\[ \de\lb \lab \cdot,\eta\rab -a\rb=\frac{1}{z} \de\lb \frac{\lab \cdot,\eta\rab}{z} -\frac{a}{z}\rb ,\]
see e.g. \cite{Wes95} for details and proofs.

\subsection{Projection Operators}

Now we recall briefly and present some properties of projection operators acting on functions of white noise. First we fix a notation as follows. For $\xi=\lb \xi_1,\ldots,\xi_d \rb \in \leb^d$ let
\[ \cP_{\xi}:=\lcb P\lb \lab \cdot, \xi_1\rab,\ldots,\lab \cdot,\xi_d \rab \rb : P \mbox{ is a polynomial }\rcb .\]
Note that the closure of $\cP_{\xi}$ in $L^p(\mu)$, $1\le p<\infty$, is given by
\begin{equation}\label{closure}
\overline{\cP_{\xi}}^{L^p(\mu)}=\lcb f\lb \lab \cdot, \xi_1\rab,\ldots,\lab \cdot,\xi_d \rab \rb :f \in L^p\lb \R^d,\mu_M;\C \rb \rcb,
\end{equation}
where $\mu_M$ is the Gaussian measure on $\R^d$ with mean zero and covariance structure $M=\lb \lab \xi_k,\xi_l\rab \rb_{k,l=1,\ldots,d}$. Also note that $\cP_{\xi} \subset \cG$.

The projection operator was first introduced in \cite{Wes95}. The basic idea of this operator is to remove the dependency on a monomial $\lab \cdot,\eta \rab$ from a random variable. This turns out to be useful to represent the pointwise product of a (generalized) random variable with Donsker's delta function. For $\eta \in \leb$ with $|\eta|=1$ let $P_{\bot,\eta}$ denote the orthogonal projection onto the orthogonal complement of $\mathrm{span}\lcb \eta \rcb$ in $L^2(\R)$ and consider its complexification, denoted by the same symbol, i.e.
\[ P_{\bot,\eta}\xi=\xi-\lab \xi,\eta\rab\eta, \quad \xi \in L^2(\R;\C). \]
It was shown in \cite[Lemma 69]{Wes95}:
\begin{lem}\label{proj1}
Let $\eta \in \mathcal{S}(\R)$ with $|\eta|=1$ and consider the unique continuous version of a smooth polynomial $\vi \in \cP$ as in (\ref{Wickpol}). Then 
\begin{equation}\label{chaosproj}
\vi\lb \cdot-\lab \cdot,\eta\rab \eta\rb=\sum_{n=0}^N\sum_{k=0}^{\left\lfloor \frac{n}{2}\right\rfloor}\frac{n!(-1)^k}{k!(n-2k)!2^k}\lab :\cdot^{\otimes(n-2k)}:,P_{\bot,\eta}^{\otimes(n-2k)}\lb \eta^{\otimes 2k}\hat{\otimes}_{2k}\vi^{(n)} \rb \rab .
\end{equation}
\end{lem}
\noindent Here $\hat{\otimes}_{2k}$ denotes the symmetrization of the contraction of tensor products, a continuous bilinear mapping $\otimes_{2k}:L^2(\R;\C)^{\otimes (2k+n)}\times L^2(\R;\C)^{\otimes (2k+m)} \to L^2(\R;\C)^{\otimes (n+m)}$ characterized by the property
\begin{align*}
&(\xi_1\otimes \cdots \otimes \xi_{2k+n})\otimes_{2k}(\zeta_1\otimes \cdots \otimes \zeta_{2k+m})\\
&=\lab \xi_1,\zeta_1\rab \cdots \lab \xi_{2k},\zeta_{2k}\rab \xi_{2k+1}\otimes \cdots \otimes \xi_{2k+n} \otimes \zeta_{2k+1} \otimes \cdots \zeta_{2k+m}
\end{align*}
for $\xi_1,\ldots,\xi_{2k+n},\zeta_1,\ldots,\zeta_{2k+m}\in L^2(\R;\C)$, see \cite{Oba94} for details. The right-hand side of (\ref{chaosproj}) is well-defined in $\cG$ for $\eta \in L^2(\R)$ with $|\eta|=1$ and allows to state
\begin{defi}
For $\eta \in \leb$ with $|\eta|=1$ the projection operator $P_{\eta}:\cP\to \cG$ is defined by
\begin{equation}\label{proj2}
P_{\eta}\vi:=\sum_{n=0}^N\sum_{k=0}^{\left\lfloor \frac{n}{2}\right\rfloor}\frac{n!(-1)^k}{k!(n-2k)!2^k}\lab :\cdot^{\otimes(n-2k)}:,P_{\bot,\eta}^{\otimes(n-2k)}\lb \eta^{\otimes 2k}\hat{\otimes}_{2k}\vi^{(n)} \rb \rab .
\end{equation}
\end{defi}

The proof of the following theorem can be found in \cite[Theorem 71]{Wes95}.
\begin{thm}
For $\eta \in \leb$ with $|\eta|=1$ there exists a unique extension of $P_{\eta}$ to a linear continuous operator $P_{\eta}:\cG \to \cG$.
\end{thm}

\begin{rem}\label{proj3}
It is obvious from (\ref{proj2}) that $\lim_{k\to \infty}\eta_k=\eta$ in the unit sphere of $\leb$ implies 
\begin{equation}\label{repr} 
\lim_{k\to \infty}P_{\eta_k}\vi=P_{\eta}\vi
\end{equation}
in $\cG$ for every fixed $\vi \in \cP$. It is even possible to show with techniques similar to those in the proof of \cite[Theorem 71]{Wes95} that for every $r\ge 0$ there exists $q\ge 0$ such that
\[ \lim_{k\to \infty}\sup_{\vi \in \cP,\| \vi\|_q\le 1}\lVb P_{\eta_k}\vi- P_{\eta}\vi\rVb_r =0,\]
i.e. we have uniform convergence.
\end{rem}

\begin{lem}\label{proj4}
For $\eta \in \leb$ with $|\eta|=1$ and $\vi,\psi \in \cG$ it holds $P_{\eta}\lb \vi \cdot \psi\rb=P_{\eta}\vi \cdot P_{\eta}\psi$.
\end{lem}
\begin{proof}
The property is clear by definition if $\eta \in \mathcal{S}(\R)$ and $\vi,\psi \in \cP$. For general $\eta \in \leb$ let $(\eta_k)_{k\in \N}$ be a sequence in $ \mathcal{S}(\R)$ converging to $\eta$ in $\leb$ and fulfilling $|\eta_k|=1$ for all $k\in \N$. By (\ref{repr}) and continuity of pointwise multiplication in $\cG$ it follows
\[ P_{\eta}\lb \vi \cdot \psi\rb=\lim_{k\to \infty}P_{\eta_k}\lb \vi \cdot \psi\rb=\lim_{k\to \infty}P_{\eta_k}\vi \cdot P_{\eta_k}\psi=P_{\eta}\vi \cdot P_{\eta}\psi \]
in $\cG$ for fixed $\vi,\psi \in \cP$. The general case $\vi,\psi \in \cG$ follows by another approximation.
\end{proof}

For $\xi=(\xi_1,\ldots,\xi_d)\in L^2(\R)^d$ we have that $\cP_{\xi} \subset \cG$, so $P_{\eta}$ is well-defined on $\cP_{\xi}$ for any $\eta\in \leb$ with $|\eta|=1$. The following lemma characterizes the action of $P_{\eta}$ on $\cP_{\xi}$.
\begin{lem}\label{polpw}
For $\eta\in \leb$ with $|\eta|=1$ and $\xi=(\xi_1,\ldots,\xi_d)\in L^2(\R)^d$, $d\in \N$, we have
\[ P_{\eta}P\lb \lab \cdot,\xi_1\rab ,\ldots, \lab \cdot,\xi_d\rab\rb=P\lb \lab \cdot,P_{\bot,\eta}\xi_1\rab ,\ldots, \lab \cdot,P_{\bot,\eta}\xi_d\rab\rb \]
for every polynomial $P$ on $\R^d$.
\end{lem}
\begin{proof}
It is clear by (\ref{proj2}) that $P_{\eta}\lab \cdot,\xi_j\rab=\lab \cdot,P_{\bot,\eta}\xi_j\rab$ for every $j=1,\ldots,d$. Then the general statement follows from Lemma \ref{proj4}.
\end{proof}

\begin{lem}\label{NM}
Let $1\le p<\infty$ and $M,N\in \R^{d\times d}$ be symmetric with $0<N\le M$, i.e. $0<x^TNx\le x^TMx$ for all $x\in \R^d\setminus \lcb 0\rcb$. Then for all $f\in L^p\lb \R^d,\mu_M;\mathbb{C}\rb$ it holds
\[ \lVb f\rVb_{L^p\lb \R^d,\mu_N;\mathbb{C}\rb} \le \lb \frac{\det M}{\det N}\rb^{1/2p} \lVb f\rVb_{L^p\lb \R^d,\mu_M;\mathbb{C}\rb} .\]
\end{lem}
\begin{proof}
Note that $0<N\le M$ implies $0<M^{-1}\le N^{-1}$. Hence
\begin{align*}
\lVb f\rVb_{L^p\lb \R^d,\mu_N;\mathbb{C}\rb}^p&=\frac{1}{\sqrt{(2\pi)^d \det N}}\int_{\R^d}|f|^p\exp\lb -\halb x^TN^{-1}x\rb \, dx\\
&\le \frac{1}{\sqrt{(2\pi)^d \det N}}\int_{\R^d}|f|^p\exp\lb -\halb x^TM^{-1}x\rb \, dx\\
 & =\sqrt{\frac{\det M}{\det N}} \lVb f\rVb_{L^p\lb \R^d,\mu_M;\mathbb{C}\rb}^p
\end{align*}
for all $f\in L^p\lb \R^d,\mu_M;\mathbb{C}\rb$.
\end{proof}

\noindent The following proposition enables us to extend $P_{\eta}$ to classes of subspaces of $L^p(\mu)$ by continuity.
\begin{prop}\label{contproj}
Let $d\in \N$ and $\eta ,\xi_1,\ldots,\xi_d\in \leb$ be linearly independent with $|\eta| =1$. Then there exists $C(\eta,\xi)\in \R$ such that for any $1\le p<\infty$ we have $\lVb P_{\eta}\vi\rVb_{L^p(\mu)}\le C(\eta,\xi) \lVb \vi \rVb_{L^p(\mu)}$ for $\vi \in \cP_{\xi}$, where $\xi=\lb \xi_1,\ldots,\xi_d\rb$. Hence, $P_{\eta}$ extends uniquely to a bounded linear operator from $\overline{\cP_{\xi}}^{L^p(\mu)}$ to $L^p(\mu)$.
\end{prop}
\begin{proof}
The matrices $M:=\lb \lab \xi_k,\xi_l\rab\rb_{k,l=1,\ldots,d}$ and $N:=\lb \lab P_{\bot,\eta}\xi_k,P_{\bot,\eta}\xi_l\rab\rb_{k,l=1,\ldots,d}$ are the covariance matrices of the Gaussian vectors $\lb \lab \cdot,\xi_k\rab\rb_{k=1,\ldots,d}$ \linebreak and $\lb \lab \cdot,P_{\bot,\eta}\xi_k\rab\rb_{k=1,\ldots,d}$, respectively. Linear independence and the fact that $\|P_{\bot,\eta}\|_{\cL(\leb)}=1$ yields $0<N\le M$. Then for $1\le p<\infty$ and a polynomial $P$ we can estimate using Lemma \ref{polpw} and Lemma \ref{NM}
\begin{align*}
&\lVb P_{\eta}P\lb \lab \cdot,\xi_1 \rab,\ldots,\lab \cdot,\xi_d \rab\rb \rVb_{L^p(\mu)}^p\\
&=\lVb P\lb  \lab \cdot,P_{\bot,\eta}\xi_1 \rab,\ldots,\lab \cdot,P_{\bot,\eta}\xi_d \rab \rb \rVb_{L^p(\mu)}^p=\lVb P\rVb_{L^p\lb \R^d,\mu_N;\C\rb}^p\\
&\le \sqrt{\frac{\det M}{\det N}} \lVb P\rVb_{L^p\lb \R^d,\mu_M;\C\rb}^p=\sqrt{\frac{\det M}{\det N}}\lVb P\lb \lab \cdot,\xi_1 \rab,\ldots,\lab \cdot,\xi_d \rab\rb \rVb_{L^p(\mu)}^p,  
\end{align*}
which shows the assertion. The fact that $C(\eta,\xi)$ can be chosen independently of $p$ can be seen by $C(\eta,\xi)=\sup_{1\le p<\infty}\lb \frac{\det M}{\det N}\rb^{1/2p}=\sqrt{\frac{\det M}{\det N}}$.
\end{proof}

The following characterizes the extension of $P_{\eta}$ provided by Proposition \ref{contproj} by generalizing Lemma \ref{polpw}.
\begin{lem}\label{projpw}
Let $\eta, \xi, p, M, N$ be as in Proposition \ref{contproj}. Then for all $f\in L^p\lb \R^d,\mu_M;\C\rb$ we have
\[  P_{\eta}f\lb \lab \cdot,\xi_1\rab ,\ldots, \lab \cdot,\xi_d\rab\rb=f\lb \lab \cdot,P_{\bot,\eta}\xi_1\rab ,\ldots, \lab \cdot,P_{\bot,\eta}\xi_d\rab\rb .\]
\end{lem}
\begin{proof}
Let $(P_n)_{n\in \N}$ be a sequence of polynomials such that $\lim_{n\to \infty}P_n=f$ in $L^p\lb \R^d,\mu_M;\C\rb$. By Lemma \ref{NM} we also have convergence in $L^p\lb \R^d,\mu_N;\C\rb$. Then Lemma \ref{polpw} and Proposition \ref{contproj} imply
\begin{align*}
 P_{\eta}f\lb \lab \cdot,\xi_1\rab ,\ldots, \lab \cdot,\xi_d\rab\rb&=\lim_{n\to \infty}P_{\eta}P_n\lb \lab \cdot,\xi_1\rab ,\ldots, \lab \cdot,\xi_d\rab\rb\\
 & =\lim_{n\to \infty}P_n\lb \lab \cdot,P_{\bot,\eta}\xi_1\rab ,\ldots, \lab \cdot,P_{\bot,\eta}\xi_d\rab\rb\\
&=f\lb \lab \cdot,P_{\bot,\eta}\xi_1\rab ,\ldots, \lab \cdot,P_{\bot,\eta}\xi_d\rab\rb
\end{align*}
in $L^p(\mu)$.
\end{proof}

\subsection{Pointwise Product with Donsker's Delta Function}

A useful formula for the pointwise product of Donsker's delta function with elements from $\cG$ is the following.
\begin{thm}\label{Wickformula}
Let $\eta \in \leb \setminus{\{0\}}$. Then
\begin{equation}\label{wickformel} 
\de\lb \lab \cdot,\eta\rab\rb \cdot \vi =\de\lb \lab \cdot,\eta\rab\rb \wick P_{\frac{\eta}{|\eta|}}\vi ,
\end{equation}
for all $\vi \in \cG$.
\end{thm}

\noindent It was discovered in \cite{GSV11} and treated systematically in \cite{Vog10}, see e.g. \cite[Theorem 4.24]{Vog10}. Unfortunately, it is not always easy to check whether a given white noise function $\vi \in L^2(\mu)$ is from $\cG$. Moreover, the representation of $P_{\frac{\eta}{|\eta|}}\vi$ from Lemma \ref{projpw} does not apply to general $\vi\in \cG$. So we need to make a refinement of this theorem which is applicable to our problem. 
\begin{thm}[Wick Formula]\label{refWick}
Let $\eta,\xi_1,\ldots,\xi_d \in L^2(\R)$ be linearly independent and set $\xi:=\lb \xi_1,\ldots,\xi_d\rb$. Then for each $1<p<\infty$ the linear operator
\begin{equation}\label{refwick1}
\cP_{\xi}\ni \vi \mapsto \de\lb \lab \cdot,\eta\rab\rb \cdot \vi  \in \cG'
\end{equation}
has a unique continuous extension to $\overline{\cP_{\xi}}^{L^p(\mu)}$. It is given by
\begin{equation}\label{refwick2}
\de\lb \lab \cdot,\eta\rab\rb \cdot \vi  =\de\lb \lab \cdot,\eta\rab\rb \wick f\lb \lab \cdot,P_{\bot,\frac{\eta}{|\eta|}}\xi_1\rab ,\ldots, \lab \cdot,P_{\bot,\frac{\eta}{|\eta|}}\xi_d\rab\rb
\end{equation}
for $\vi=f\lb \lab \cdot,\xi_1\rab ,\ldots, \lab \cdot,\xi_d\rab\rb \in \overline{\cP_{\xi}}^{L^p(\mu)}$.
\end{thm}
\begin{proof}
The operator $P_{\frac{\eta}{|\eta|}}$ is continuous from $\overline{\cP_{\xi}}^{L^p(\mu)}$ to $L^p(\mu)$, which is continuously embedded in $\cG'$, and the Wick product acts continuously from $\cG' \times \cG'$ to $\cG'$. Hence, the existence of a unique extension follows from Theorem \ref{Wickformula} and density of $\cP_{\xi}$ in $\overline{\cP_{\xi}}^{L^p(\mu)}$. Then (\ref{refwick2}) follows from Lemma \ref{projpw}.
\end{proof}

\begin{rem}
The condition on linear independence of $\eta,\xi_1,\ldots,\xi_d$ in Theorem \ref{refWick} (and also in Proposition \ref{contproj} and Lemma \ref{projpw} before) can actually be relaxed to the condition $\eta \notin \mathrm{span}\lcb \xi_1,\ldots,\xi_d\rcb$. This follows from the fact that there exists $m\in \N$ and $\zeta=(\zeta_1,\ldots,\zeta_m)\in \leb^m$ such that $\zeta_1,\ldots,\zeta_m$ is a linear basis  of $\mathrm{span}\lcb \xi_1,\ldots,\xi_d\rcb$ and thus we have that $\eta,\zeta_1,\ldots,\zeta_d$ is linearly independent with $\cP_{\zeta}=\cP_{\xi}$.
\end{rem}

\begin{exam}\label{BMBB}
If $0\le t\le T$, $T\neq 0$, and $\eta=\frac{1}{\sqrt{T}}\1_{[0,T)}$, then $P_{\bot,\eta}\1_{[0,t)}=\1_{[0,t)}-\frac{t}{T}\1_{[0,T)}$. Hence
\[ P_{\eta}B_t=\lab \cdot , P_{\bot,\eta}\1_{[0,t)}\rab=\lab \cdot , \1_{[0,t)}-\frac{t}{T}\1_{[0,T)}\rab =B_t-\frac{t}{T}B_T=:X_t .\]
Since $(B_t)_{t\in [0,T]}=\lb \lab \cdot , \1_{[0,t)}\rab \rb_{t\in [0,T]}$ is a standard Brownian motion (starting at $0$), $(X_t)_{t\in [0,T]}$  is a Brownian bridge starting and ending in $0$. Thus it follows from Theorem \ref{refWick} applied to $\xi_j=\1_{[0,t_j)}$ that for fixed $0<t_1<\cdots<t_d<T$ we have
\[ \de(B_T)\cdot f(B_{t_1},\ldots,B_{t_d})=\de(B_T)\wick f(X_{t_1},\ldots,X_{t_d}),\]
for any measurable $f:\R^d\to \C$ with $f(B_{t_1},\ldots,B_{t_d})\in L^p(\mu)$ for some $1<p<\infty$.
\end{exam}

\section{Self-intersection Local Time of the Brownian Bridge}

Let $B=(B_t)_{t\in [0,T]}=\lb \lab \cdot, \1_{[0,t)}\rab\rb_{t\in [0,T]}$, $0<T<\infty$, i.e. a one-dimensional standard Brownian motion, see Example \ref{BMBB}. Let $a,b\in \R$ and let us consider
\begin{equation}
X_t:=a\lb 1-\frac{t}{T}\rb +b\frac{t}{T}+ B_t-\frac{t}{T}B_T=a\lb 1-\frac{t}{T}\rb +b\frac{t}{T}+\lab \cdot, \1_{[0,t)}-\frac{t}{T}\1_{[0,T)}\rab
\end{equation}
for $0\le t\le T$, i.e. a one-dimensional Brownian bridge from $a$ to $b$ on $[0,T]$ on the white noise space. It can be verified easily that $(X_t)_{t\in [0,T]}$ is a Gaussian process with mean function 
\[
\E(X_t):=\int_{\td}a\lb 1-\frac{t}{T}\rb +b\frac{t}{T}+\lab \om, \1_{[0,t)}-\frac{t}{T}\1_{[0,T)}\rab \, d\mu(\om)=a\lb 1-\frac{t}{T}\rb +b\frac{t}{T}
\]
for $0\le t\le T$ and covariance function
\begin{align*}\label{BB3}
\mathrm{cov}(X_s,X_t)&:=\int_{\td}\lab \om, \1_{[0,s)}-\frac{s}{T}\1_{[0,T)}\rab\, \lab \om, \1_{[0,t)}-\frac{t}{T}\1_{[0,T)}\rab \, d\mu(\om)\\
&=s\wedge t-\frac{st}{T}=\frac{s\wedge t}{T} \lb T-(s\vee t) \rb, \quad 0\le s,t\le T . 
\end{align*}
In the following we define the variance of $X_t$ as $\var(X_t):=\mathrm{cov}(X_t,X_t)$. We define self-intersection local times of Brownian bridge during the time interval $[0,T]$ by
\begin{equation}\label{BB7}
\mathcal{I}^{BB}:=\int_0^T\int_0^t \de(X_t-X_s)\, ds dt,
\end{equation}
where $\de$ denotes the Dirac delta distribution at $0$. $\mathcal{I}^{BB}$ is interpreted as the amount of time the sample path of Brownian bridge $X$ spends intersect itself within the time interval $[0,T]$. It can be proved using the characterization of Hida distributions and the analysis in Hida spaces that for any spatial dimension (the renormalized) $\mathcal{I}^{BB}$ exists as a Hida distribution. The space of Hida distributions is larger than $\cG'$, see \cite{PT95} for more information. For dimension one it is also possible to give mathematically rigorous meaning to $\mathcal{I}^{BB}$ as a square-integrable function by using an approximation procedure. One common way to do this is by approximating the Dirac delta distribution. More precisely, we interpret ($\ref{BB7}$) as the limit of the approximated self-intersection local times $\mathcal{I}_{\ep}^{BB}$ of a one-dimensional Brownian bridge $X$ defined by
\[ \mathcal{I}_{\ep}^{BB} :=\int_0^T\int_0^t p_{\ep}(X_t-X_s)\, ds\, dt, \quad \ep >0, \]
as $\ep \to 0$, where $p_{\ep}$ is the heat kernel given by
\[ p_{\ep}(x)=\frac{1}{\sqrt{2\pi \ep}}\exp\lb -\frac{x^2}{2\ep}\rb, \quad x\in \R. \]

\begin{thm}\label{SILT1DBB}
The approximated self-intersection local time $\mathcal{I}_{\ep}^{BB}$ of a one-dimensional Brownian bridge $X$ converges in $L^2(\mu)$ as $\ep$ tends to zero, i.e.
\[ \lim_{\ep \downarrow 0}\mathcal{I}_{\ep}^{BB} =:\mathcal{I}^{BB}\in L^2(\mu). \] 
\end{thm}
\begin{proof}
We observe that
\begin{align*}
\mathcal{I}_{\ep}^{BB} & = \int_0^T\int_0^tp_{\ep}(X_t-X_s)\, ds\, dt\\
 & = \frac{1}{2\pi}\int_0^T\int_0^t\int_\R \exp\lb i\xi(X_t-X_s)\rb \exp\lb -\frac{\ep}{2}|\xi|^2\rb \, d\xi \, ds \, dt.
\end{align*}
Let us denote $D:=\lcb (s_1,t_1,s_2,t_2):0<s_1<t_1<T \, , 0<s_2<t_2<T \rcb$.
Hence,
\begin{align*}
\E\lb (\mathcal{I}_{\ep}^{BB})^2\rb&=\E\lb \frac{1}{4\pi^2}\int_D \int_{\R^2} \exp\lb i \sum_{j=1}^2\xi_j(X_{t_j}-X_{s_j})\rb \exp\lb -\frac{\ep}{2}\sum_{j=1}^2\xi_j^2\rb d\xi \, ds \, dt\rb\\
&=\frac{1}{4\pi^2} \int_D \int_{\R^2} \E\lb \exp\lb i \sum_{j=1}^2\xi_j(X_{t_j}-X_{s_j})\rb \rb \exp\lb -\frac{\ep}{2}\sum_{j=1}^2\xi_j^2\rb d\xi \, ds \, dt\\
&=\frac{1}{4\pi^2}\int_D \int_{\R^2} \exp\lb i\E\lb \sum_{j=1}^2\xi_j(X_{t_j}-X_{s_j})\rb -\frac{1}{2}\var\lb \sum_{j=1}^2\xi_j(X_{t_j}-X_{s_j})\rb \rb \\
&\qquad \times \exp\lb -\frac{\ep}{2}\sum_{j=1}^2\xi_j^2\rb d\xi \, ds \, dt ,
\end{align*}
where we use that $X_{t_j}-X_{s_j}$, $j=1,2$, are Gaussian random variables. Note that by Lebesgue's dominated convergence theorem $\E\lb (\mathcal{I}_{\ep}^{BB})^2\rb$ converges to
\[ \beta_2:=\frac{1}{4\pi^2}\int_D \int_{\R^2} \exp\lb i\E\lb \sum_{j=1}^2\xi_j(X_{t_j}-X_{s_j})\rb -\frac{1}{2}\var\lb \sum_{j=1}^2\xi_j(X_{t_j}-X_{s_j})\rb \rb  d\xi \, ds \, dt\]
as $\ep$ tends to zero, provided 
\[ \al_2:=\frac{1}{4\pi^2}\int_D \int_{\R^2} \exp\lb -\frac{1}{2}\var \lb \sum_{j=1}^2\xi_j(X_{t_j}-X_{s_j})\rb \rb  d\xi \, ds \, dt <\infty .\]
We also consider
\begin{align*} 
\E\lb \mathcal{I}_{\ep}^{BB} \mathcal{I}_{\de}^{BB} \rb&=\frac{1}{4\pi^2}\int_D \int_{\R^2} \exp\lb i\E\lb \sum_{j=1}^2\xi_j(X_{t_j}-X_{s_j})\rb -\frac{1}{2}\var\lb \sum_{j=1}^2\xi_j(X_{t_j}-X_{s_j})\rb \rb \\
&\qquad \times \exp\lb -\frac{\ep}{2}\xi_1^2 -\frac{\de}{2}\xi_2^2\rb d\xi \, ds \, dt.
\end{align*}
If $\al_2<\infty$, then we also have that
\[ \lim_{(\ep ,\de) \to (0,0)} \E\lb \mathcal{I}_{\ep}^{BB} \mathcal{I}_{\de}^{BB} \rb=\beta_2. \]
Moreover, this implies that $\lb \mathcal{I}_{\ep}^{BB}\rb_{\ep >0}$ converges in $L^2(\mu)$ as $\ep$ tends to zero. Indeed, we show that $\lb \mathcal{I}_{\ep}^{BB}\rb_{\ep >0}$ is a Cauchy sequence in $L^2(\mu)$: Let $\gamma >0$, then there exists $N\in \N_0$ such that for all $0<\ep ,\de \le \frac{1}{N}$
\[
\E\lb (\mathcal{I}_{\ep}^{BB} -\mathcal{I}_{\de}^{BB})^2\rb=\E\lb(\mathcal{I}_{\ep}^{BB})^2 \rb +\E\lb (\mathcal{I}_{\de}^{BB})^2\rb-2\E\lb \mathcal{I}_{\ep}^{BB} \mathcal{I}_{\de}^{BB} \rb<\gamma .
\]
Therefore, for symmetry reason it is sufficient to show that
\[ \gam_2 :=\int_{D'} \int_{\R^2}\exp\lb -\frac{1}{2}\var \lb \sum_{j=1}^2\xi_j(X_{t_j}-X_{s_j})\rb \rb  d\xi \, ds \, dt  \]
is finite, where $D':=D\cap \{ t_1<t_2 \}$. Furthermore we decompose $D'$ into three disjoint sets, i.e. $D'=D_1\sqcup D_2 \sqcup D_3$ where
\begin{align*}
D_1&:=\{(s_1,t_1,s_2,t_2):0<s_1<t_1<s_2<t_2<T\},\\
D_2&:=\{(s_1,t_1,s_2,t_2):0<s_1<s_2<t_1<t_2<T\},\\
D_3&:=\{(s_1,t_1,s_2,t_2):0<s_2<s_1<t_1<t_2<T\}.
\end{align*}
We show that
\[ \gam_2^l:=\int_{D_l}\int_{\R^2}\exp\lb-\frac{1}{2}\var\lb \sum_{j=1}^2\xi_j(X_{t_j}-X_{s_j})\rb \rb\,d\xi\,ds \, dt \]
is finite for $l=1,2,3$. Computing the Gaussian integral we get that
\[ \gam_2^l=2\pi\sqrt{T}\int_{D_l}\lb(t_1-s_1)(t_2-s_2)\lb T-(t_1-s_1)-(t_2-s_2)+2m_l\rb-Tm_l^2\rb^{-1/2}\, ds \, dt \]
where $m_l=m_l(s_1,t_1,s_2,t_2)=d x\lb[s_1,t_1]\cap[s_2,t_2]\rb$ is the length of the intersection of $[s_1,t_1]$ and $[s_2,t_2]$, i.e. $m_1=0$, $m_2=t_1-s_2$, and $m_3=t_1-s_1$. To get an estimate on $\gam_2^l$ in each case $l=1,2,3$ we use the following fact: Let $p\colon\R\to\R$ be an arbitrary polynomial of degree $2$ with leading coefficient $-1$ and let $I\subset\R$ be an interval such that $p(x)\ge 0$ for $x\in I$. Then
\[ \int_Ip(x)^{-1/2}\,dx\le\pi. \]
This follows from $\int_{-1}^1(1-x^2)^{-1/2}\,dx=\pi$. Let us first consider $l=1$. Note that for all $t_1,s_2,t_2$ we have
\[ \int_0^{t_1}\lb(t_1-s_1)(t_2-s_2)\lb T-(t_1-s_1)-(t_2-s_2)\rb\rb^{-1/2}\,ds_1\le\pi(t_2-s_2)^{-1/2}. \]
Hence
\[ \gam_2^1 \le 2\pi\sqrt{T}\int_0^T\int_0^{t_2}\int_0^{s_2}\pi(t_2-s_2)^{-1/2}\,dt_1\,ds_2\,dt_2<\infty .\]
Now we proceed for $l=2$. For all $s_2,t_1,t_2$ it holds
\begin{align*}
&\int_0^{s_2}\lb(t_1-s_1)(t_2-s_2)\lb T-(t_1-s_1)-(t_2-s_2)+2(t_1-s_2)\rb-T(t_1-s_2)^2\rb^{-1/2}\,ds_1\\
&\le \pi(t_2-s_2)^{-1/2}
\end{align*}
and thus
\[ \gam_2^2\le 2\pi\sqrt{T}\int_0^T\int_0^{t_2}\int_0^{t_1}\pi(t_2-s_2)^{-1/2}\,ds_2\,dt_1\,dt_2<\infty. \]
Finally we check for $l=3$. For all $s_1,t_1,t_2$ it holds
\begin{align*}
&\int_0^{s_1}\lb(t_1-s_1)(t_2-s_2)\lb T-(t_1-s_1)-(t_2-s_2)+2(t_1-s_1)\rb-T(t_1-s_1)^2\rb^{-1/2}\,ds_2\\
&\le \pi(t_1-s_1)^{-1/2}
\end{align*}
which yields
\[ \gam_2^3\le 2\pi\sqrt{T}\int_0^T\int_0^{t_2}\int_0^{t_1}\pi(t_1-s_1)^{-1/2}\,ds_1\,dt_1\,dt_2<\infty. \]
As a conclusion, we have $\gam_2=\gam_2^1+\gam_2^2+\gam_2^3<\infty$ and the proof is finished.
\end{proof}

\begin{cor}\label{expSILT1DBB}
For $z\in \C$ with $\operatorname{Re}z\le 0$ and $1\le p<\infty$ it holds that $\exp\lb z \mathcal{I}_{\ep}^{BB} \rb$ converges to $\exp\lb z\mathcal{I}^{BB} \rb$ in $L^p(\mu)$ as $\ep \to 0$.
\end{cor}
\begin{proof}
Since $\mathcal{I}_{\ep}^{BB}$ converges to $\mathcal{I}^{BB}$ in $L^2(\mu)$, then $\mathcal{I}_{\ep}^{BB}$ also converges to $\mathcal{I}^{BB}$ in probability. Using the continuity of $x\mapsto \exp\lb zx\rb$, we have $\exp\lb z\mathcal{I}_{\ep}^{BB}\rb$ converges in probability to $\exp\lb z\mathcal{I}^{BB}\rb$. It is clear that $\lvb \exp\lb z\mathcal{I}_{\ep}^{BB} \rb \rvb \le 1$ for all $\ep >0$, since $\mathcal{I}_{\ep}^{BB} >0$ for all $\ep >0$. Therefore, by using a dominated convergence theorem (see e.g. \cite[Theorem 17.4]{JP04}) we can conclude that $\exp\lb z\mathcal{I}_{\ep}^{BB}\rb$ converges to $\exp\lb z\mathcal{I}^{BB}\rb$ in $L^p(\mu)$ as $\ep \to 0$ for all $1\le p<\infty$.
\end{proof}

By using the integral decomposition method as above we can establish a proof for the $L^2(\mu)$-approximation of self-intersection local time of a one-dimensional Brownian motion. The proof is almost identical to that of Theorem \ref{SILT1DBB} and even simpler due to the independence of increments of Brownian motion. Hence, we state the results without details and proofs.
\begin{thm}\label{SILT1DBM}
The approximated self-intersection local time
\[ \mathcal{I}_{\ep}^{BM}:=\int_0^T\int_0^tp_{\ep}(B_t-B_s)\, ds\, dt, \quad \ep >0, \]
of a one-dimensional Brownian motion $B$ converges in $L^2(\mu)$ as $\ep$ tends to zero, i.e.
\[ \lim_{\ep \downarrow 0}\mathcal{I}_{\ep}^{BM} =:\mathcal{I}^{BM}\in L^2(\mu). \] 
\end{thm}

\begin{cor}\label{expSILT1DBM}
For $z\in \C$ with $\operatorname{Re}z\le 0$ and $1\le p<\infty$ it holds that $\exp\lb z \mathcal{I}_{\ep}^{BM} \rb$ converges to $\exp\lb z\mathcal{I}^{BM}\rb$ in $L^p(\mu)$ as $\ep \to 0$.
\end{cor}

\noindent The limit object $\mathcal{I}^{BM}$ in Theorem \ref{SILT1DBM} is called a one-dimensional self-intersection local time of Brownian motion, and is usually denoted by
\begin{equation}\label{SILT1}
\int_0^T\int_0^t\de(B_t-B_s)\, ds\, dt,
\end{equation}
where $\de$ is the Dirac-delta distribution at $0$. De Faria et al in \cite{FHSW97} proved that for any spatial dimension of the Brownian motion $B$ self-intersection local time $\int_0^T\int_0^t\de(B_t-B_s)\, ds\, dt$, after suitably renormalized, exists as a Hida distribution.

\section{Feynman Integrand for Electrons in Random Media}

Recall from the introduction that from the Gaussian scattering potential for electrons in random media we can obtain informally the corresponding Feynman integrand (with Dirac delta correlation function) without kinetic energy part as
\begin{equation}\label{integrand1}
\exp\lb \int_0^T\lb -\frac{ik}{2\hbar} \int_0^T \de\lb x(t)-x(s)\rb ds\,\rb dt \rb \cdot \de_{x_T}\lb x(T)\rb,
\end{equation}
see (\ref{motiv}). The Donsker's delta function here is used to pin the endpoint of the paths. Now we set $g:=\frac{k}{2\hbar}>0$. To get the Feynman integrand with kinetic energy we follow the complex-scaling ansatz proposed by Cameron \cite{Cam60} and Doss \cite{Doss82}, i.e., we multiply all Brownian motion by $\sqrt{i}$ and obtain the informal product
\begin{equation}\label{integrand2}
\exp\lb \frac{1}{i} \int_0^T \lb -ig \int_0^T \de\lb \sqrt{i}(B_t-B_s)\rb ds\rb dt \rb \cdot \de_{x_T}\lb x_0+ \sqrt{i} B_T\rb,
\end{equation}
where $(B_t)_{t\in [0,T]}$ is a standard Brownian motion. We call the expression (\ref{integrand2}) Feynman-Kac-Cameron-Doss integrand. Recall two sequences approximating Dirac delta distribution:
\[ p_{\ep}(x)=\frac{1}{\sqrt{2\pi \ep}}\exp\lb -\frac{x^2}{2\ep}\rb \quad \mbox{ and } \quad q_{\ep}(x)=\frac{1}{\sqrt{2\pi i\ep}}\exp\lb -\frac{x^2}{2i\ep}\rb ,\quad \ep >0,\]
i.e., the heat kernel and the free Schr\"odinger kernel, respectively. It is easy to see that $q_{\ep}(\sqrt{i}x)=\frac{1}{\sqrt{i}}p_{\ep}(x)$. 
Now we define the first factor in the product above in the following sense:
\begin{align*}
&\exp\lb \frac{1}{i} \int_0^T \lb -ig \int_0^T \de\lb \sqrt{i}(B_t-B_s)\rb ds\rb dt \rb\\
&:=\lim_{\ep \to 0}\exp\lb \frac{1}{i} \int_0^T \lb -ig \int_0^T q_{\ep}\lb \sqrt{i}(B_t-B_s)\rb ds\rb dt \rb\\
&=\lim_{\ep \to 0}\exp\lb-g  \int_0^T \int_0^T \frac{1}{\sqrt{i}}p_{\ep}\lb B_t-B_s\rb ds\, dt \rb\\
&=\lim_{\ep \to 0}\exp\lb-g i^{-1/2}  \int_0^T \int_0^T p_{\ep}\lb B_t-B_s\rb ds\, dt \rb\\
&=\exp\lb-g i^{-1/2}  \int_0^T \int_0^T \de\lb B_t-B_s\rb ds\, dt \rb \in L^2(\mu),
\end{align*}
by Corollary \ref{expSILT1DBM}. We always consider $i^{-1/2}$ with $\mathrm{Re}(i^{-1/2})\ge 0$. On the other hand we know that
\[ \de_{x_T}\lb x_0+ \sqrt{i} B_T\rb \in \cG' , \]
see e.g. \cite{GSV11}. Therefore we arrive at the problem of multiplication of a square-integrable function with a regular distribution. In the following we are able to give a rigorous meaning to this product as a limit object in $\cG'$. To this end we apply  the refinement of the Wick formula, i.e. Theorem \ref{refWick}. To proceed further we restrict ourselves in the special case $x_0=x_T$. This case is of particular interest from the physical application point of view. For example in the investigation of the density of states of electrons in random media. The density of states is obtained by taking Fourier transform with respect to time of the trace (diagonal element) of the electron's propagator. For more information we refer to \cite{EG64, Sam74, KL86}.

Let us fix the following notations
\[ \Phi:=\exp\lb-g i^{-1/2}  \int_0^T \int_0^T \de\lb B_t-B_s\rb ds\, dt \rb \in L^2(\mu) \]
and
\[ \Psi:=\exp\lb-g i^{-1/2}  \int_0^T \int_0^T \de\lb X_t-X_s\rb ds\, dt \rb \in L^2(\mu) ,\]
where $X_t=x_0+B_t-\frac{t}{T}B_T$, see Example \ref{BMBB}. Moreover, for $\ep >0$ and $n\in \N$ we define
\begin{align*}
\Phi_{\ep}&:=\exp\lb -gi^{-1/2} \int_0^T \int_0^T p_{\ep}\lb B_t-B_s\rb ds\, dt \rb ,\\
\Phi_{\ep,n}&:=\exp\lb -gi^{-1/2} \lb \frac{T}{n}\rb^2 \sum_{k,l=1}^n p_{\ep}\lb B_{t_k}-B_{s_l}\rb \rb ,\\
\Psi_{\ep}&:=\exp\lb -gi^{-1/2} \int_0^T \int_0^T p_{\ep}\lb X_t-X_s\rb ds\, dt \rb , \quad \mbox{ and }\\
\Psi_{\ep,n}&:=\exp\lb -gi^{-1/2} \lb \frac{T}{n}\rb^2 \sum_{k,l=1}^n p_{\ep}\lb X_{t_k}-X_{s_l}\rb \rb ,
\end{align*}
where $\lcb t_1,t_2,\ldots,t_n \rcb$ and $\lcb s_1,s_2,\ldots,s_n \rcb$ are two partitions of the interval $[0,T]$. Note that $\Phi_{\ep,n}$ and $\Psi_{\ep,n}$ are continuous square-integrable as functions of  $B_{t_k}-B_{s_l}$, $k,l=1,\ldots,n$ and  $X_{t_k}-X_{s_l}$, $k,l=1,\ldots,n$, respectively. I.e. $\Phi_{\ep,n}$ and $\Psi_{\ep,n}$ depend on Brownian motion and Brownian bridge at $n^2$ time points, respectively. Recall that $(B_t)_{t\in [0,T]}$ and $(X_t)_{t\in [0,T]}$ have continuous paths. Hence, for $n\to \infty$ the functions $\Phi_{\ep,n}$ and $\Psi_{\ep,n}$ converge $\mu$-a.s. to $\Phi_{\ep}$ and $\Psi_{\ep}$, respectively (approximation of the Riemann integral by a Riemann sum). Thus, we also have convergences in $L^2(\mu)$ by Lebesgue's dominated convergence theorem by using the uniform upper bound equals one. Let us fix $\ep >0$ and $n\in \N$. Denote also $\eta:=\frac{\1_{[0,T)}}{\sqrt{T}}$. Since Donsker's delta function is homogeneous of degree $-1$ and by using Theorem \ref{refWick} and Example \ref{BMBB} we have
\[
\Phi_{\ep,n} \cdot \de_{x_T}\lb x_0+\sqrt{i} B_T\rb
=\Psi_{\ep,n} \wick \frac{1}{\sqrt{i}}\de_{\frac{x_T-x_0}{\sqrt{i}}}\lb B_T\rb \in \cG'.
\]
Now using the $L^2(\mu)$-convergence of $\Phi_{\ep,n}$ and $\Psi_{\ep,n}$ to $\Phi_{\ep}$ and $\Psi_{\ep}$, respectively, as $n\to \infty$, and using the continuity of Wick product from $L^2(\mu)\times \cG'$ to $\cG'$ we can further define
\begin{align*}
\Phi_{\ep} \cdot \de_{x_T}\lb x_0+\sqrt{i} B_T\rb &:=\lim_{n\to \infty} \lb \Phi_{\ep,n} \cdot \de_{x_T}\lb x_0+\sqrt{i} B_T\rb \rb
=\lim_{n\to \infty}\lb \Psi_{\ep,n} \, \wick \, \frac{1}{\sqrt{i}}\de_{\frac{x_T-x_0}{\sqrt{i}}}\lb B_T\rb \rb\\
&=\lb \lim_{n\to \infty}\Psi_{\ep,n} \rb \wick \, \frac{1}{\sqrt{i}}\de_{\frac{x_T-x_0}{\sqrt{i}}}\lb B_T\rb
=\Psi_{\ep} \, \wick \, \frac{1}{\sqrt{i}}\de_{\frac{x_T-x_0}{\sqrt{i}}}\lb B_T\rb \in \cG' .
\end{align*}
As final step, using Corollary \ref{expSILT1DBM}, Corollary \ref{expSILT1DBB}, and continuity of Wick product from $L^2(\mu)\times \cG'$ to $\cG'$ we can make the following definition
\begin{align*}
\Phi \cdot \de_{x_T}\lb x_0+\sqrt{i} B_T\rb
&:=\lim_{\ep \to 0} \lb \Phi_{\ep} \cdot \de_{x_T}\lb x_0+ \sqrt{i} B_T\rb \rb
=\lim_{\ep \to 0} \lb \Psi_{\ep} \wick \frac{1}{\sqrt{i}}\de_{\frac{x_T-x_0}{\sqrt{i}}}\lb B_T\rb \rb\\
&= \lb \lim_{\ep \to 0} \Psi_{\ep}  \rb \wick \frac{1}{\sqrt{i}}\de_{\frac{x_T-x_0}{\sqrt{i}}}\lb B_T\rb
=\Psi \wick \frac{1}{\sqrt{i}}\de_{\frac{x_T-x_0}{\sqrt{i}}}\lb B_T\rb.
\end{align*}
Since $\Psi \in L^2(\mu)\subset \cG'$ and $\de_{x_T}\lb x_0+\sqrt{i} B_T\rb=\frac{1}{\sqrt{i}}\de_{\frac{x_T-x_0}{\sqrt{i}}}\lb B_T\rb \in \cG'$, we have given a meaning to the product (\ref{Doss}) as an element of $ \cG'$. We summarize our main result in the following theorem.
\begin{thm}
The Feynman-Kac-Cameron-Doss integrand of the electrons in random media with non-local Dirac delta action
\[ \exp\lb \frac{1}{i} \int_0^T \lb -\frac{ik}{2\hbar} \int_0^T \de\lb \sqrt{i}(B_t-B_s)\rb ds\rb dt \rb \cdot \de_{x_T}\lb x_0+\sqrt{i}B_T\rb, \]
where $(B_t)_{t\in [0,T]}$ is a one-dimensional standard Brownian motion and $x_0=x_T\in \R$, is a regular distribution of white noise, i.e. an element of $\cG'$. Furthermore, it holds that
\begin{align*}
&\exp\lb \frac{1}{i} \int_0^T \lb -\frac{ik}{2\hbar} \int_0^T \de\lb \sqrt{i}(B_t-B_s)\rb ds\rb dt \rb \cdot \de_{x_T}\lb x_0+\sqrt{i}B_T\rb\\
&=\exp\lb-\frac{k}{2\hbar} i^{-1/2}  \int_0^T \int_0^T \de\lb X_t-X_s\rb ds\, dt \rb \wick \frac{1}{\sqrt{i}}\de_{\frac{x_T-x_0}{\sqrt{i}}}\lb B_T\rb,
\end{align*}
$(X_t)_{t\in [0,T]}$ is a one-dimensional Brownian bridge given by $X_t=x_0+B_t-\frac{t}{T}B_T$.

\end{thm}
In other words we show that for the limiting case $L\to 0$ of correlation length in the Edwards model (as we mentioned in Section 1), the corresponding Feynman integrand for identic start and end point is a well-defined object as a regular distribution of white noise. Its generalized expectation gives the corresponding Feynman propagator
\begin{align}\label{propa}\nonumber
G_{\de}&=K_{\de}(x_T,T;x_0,0)\nonumber \\
&=\E_{\mu} \lb \exp\lb \frac{1}{i} \int_0^T \lb -ig \int_0^T \de\lb \sqrt{i}(B_t-B_s)\rb ds\rb dt \rb \cdot \de_{x_T}\lb x_0+ \sqrt{i} B_T\rb \rb \nonumber\\
&=\E_{\mu} \lb \exp\lb-g i^{-1/2}  \int_0^T \int_0^T \de\lb X_t-X_s\rb ds\, dt \rb \wick  \de_{x_T}\lb x_0+ \sqrt{i} B_T\rb \rb \nonumber\\
&=\frac{1}{\sqrt{2\pi i T}}\exp\lb -\frac{1}{2iT}(x_T-x_0)^2\rb \,\, \E_{\mu} \lb \exp\lb-g i^{-1/2}  \int_0^T \int_0^T \de\lb X_t-X_s\rb ds\, dt \rb \rb ,
\end{align}
where $x_T=x_0$ and $g=\frac{k}{2\hbar}$. 

\section{Conclusion}

Using the explicit formula for the Feynman propagator $K_{\de}(x_0,T;x_0,0)$, see (\ref{propa}), one can study an important physical object, namely, the density of states. It can be represented as a function of the energy of an electron in random media by taking the Fourier transform of $K_{\de}(x_0,T;x_0,0)$ in the time variable $T$. The analysis of density of states in disordered structure has been a main object of interest in \cite{EG64} and \cite{Sam74}.

We also would like to mention the interesting research on random Hamiltonians with point interaction. An excellent reference on this subject is the monograph of Albeverio et al \cite{AGHH88}. This topic concerns with Schr\"odinger operator with stochastic potential and has been used for models of amorphous solids and disordered system with point (Dirac delta) interactions. These models are different from the Edwards model for electrons in random media considered in the present paper, although these models share some commons features. See also the remarks in the Notes of Chapter III in \cite{AGHH88}. In particular, the physical formulations are different as we briefly indicate below. In the random Hamiltonian model the Dirac delta potential as well as the stochasticity are incorporated right from beginning in the Schr\"odinger representation of the model. More precisely, the Hamiltonian $H_{\omega}$ is of the form
\[ H_{\omega}=-\frac{\hbar}{2m}\Delta+\sum_{j\in J}\eta_j \de_{r_j(\omega)}(\cdot), \]
where $J$ is an discrete index set, $r_j$ is a point source location which is a random variable defined on a probability space $(\Omega, \mathcal{F},P)$ and  $\eta_j$ is a coupling constant attached to the point source located at $r_j(\omega)$, $\omega \in \Omega$. Starting from this model one studies the properties of $H_{\omega}$ such as self-adjointness, spectrum, eigenfunctions, resonances, and scattering quantities, see Chapter III in \cite{AGHH88}. This is in contrast with the Edwards model which invokes the randomness through averaging the Feynman path integral representation of the system over all possible configurations of the point sources (scatterers). After that, the Dirac delta function in the potential part is obtained by taking limit $L\to 0$ of the correlation length. Hence, the classical action in present consideration is given by
\[ S_{\de}(x)= \int_0^T \lb \frac{m}{2}\dot{x}(t)^2 + \frac{ik}{2\hbar} \int_0^T\de(x(t)-x(s))\, ds\rb \, dt, \]
and we can read off the potential to be a complex-valued function (with negative imaginary part). This type of potentials is commonly used in the study of quantum mechanical system with unstable particles, see e.g. \cite{Wri84} and \cite{GS98}. Mathematically speaking, the corresponding Hamiltonian is no longer Hermitian, and consequently the time evolution operator is not unitary. In addition, the Dirac delta potential is non-local in time in the sense that the interaction at time $s$ still has some effects at another later time $t$. This non-local action corresponds to a non-Markov process and has self-attraction effect on the particle, see e.g. \cite{KL86}. These facts make the model considered in the present paper and its mathematical treatment substantially different from that in \cite{AGHH88}.

\vskip0.8cm
\textbf{Acknowledgment}: We thank Ludwig Streit for posing this interesting problem and Jose Luis da Silva for helpful remarks. The second author thanks the ISGS Kaiserslautern for the financial support in the form of a fellowship of the German state Rhineland-Palatinate.
The third author thanks the DAAD (German Academic Exchange Service) for a scholarship on the PhD Programme "Mathematics in Industry and Commerce" at University of Kaiserslautern and also the ISGS Kaiserslautern for the financial support in the form of a fellowship of the German state Rhineland-Palatinate. We would like to thank an unknown referee for helpful comments to the manuscript which lead us to a careful comparison of the model considered in the present paper with models using random Hamiltonians with point interactions.

\bibliographystyle{is-alpha}

\bibliography{paper}

\end{document}